\documentclass[copyright,creativecommons]{eptcs}
\providecommand{\event}{QPL 2018}

\usepackage{amsthm}
\usepackage{amsmath,amssymb,url,tikz,enumerate,rotating,color,hyperref}
\usepackage[strings]{underscore}

\numberwithin{equation}{section}
\theoremstyle{plain}
\newtheorem{theorem}[equation]{Theorem}
\newtheorem{lemma}[equation]{Lemma}
\newtheorem{proposition}[equation]{Proposition}

\theoremstyle{definition}
\newtheorem{definition}[equation]{Definition}
\newtheorem{example}[equation]{Example}
\newtheorem{remark}[equation]{Remark}

\usetikzlibrary{matrix}

\newcommand{\cat}[1]{\ensuremath{\mathbf{#1}}}
\newcommand{\op}{\ensuremath{{}^{\text{op}}}}
\newcommand{\id}[1][]{\ensuremath{\mathrm{id}_{#1}}} 
\DeclareMathOperator{\supp}{supp}
\newcommand{\ie}{\text{i.e.}\ }
\newcommand{\eg}{\text{e.g.}\ }

\title{Categories of Empirical Models}
\author{Martti Karvonen
\institute{School of Informatics\\ University of Edinburgh}
\email{martti.karvonen@ed.ac.uk}
}

\def\titlerunning{Categories of Empirical Models}
\def\authorrunning{M. Karvonen}

\begin{document}

\maketitle

\begin{abstract}
A notion of morphism that is suitable for the sheaf-theoretic approach to contextuality is developed, resulting in a resource theory for contextuality. The key features involve using an underlying relation rather than a function between measurement scenarios, and allowing for stochastic mappings of outcomes to outcomes. This formalizes an intuitive idea of using one empirical model to simulate another one with the help of pre-shared classical randomness. This allows one to reinterpret concepts and earlier results in terms of morphisms. Most notably: non-contextual models are precisely those allowing a morphism from the terminal object; contextual fraction is functorial; Graham-reductions induce morphisms, reinterpreting Vorob'evs theorem; contextual models cannot be cloned.
\end{abstract}

\section{Introduction}

In the Abramsky-Brandenburger approach to contextuality~\cite{abramsky_brandenburger_sheaf-theoretic_structure} questions of non-locality and contextuality are formulated in the language of sheaf theory. The objects defined and studied in that approach are called empirical models. In this work, we present two notions of a morphism between empirical models, allowing us to further advance the use of categorical language in the study of contextuality. Our notion of a morphism can be seen as a formalization of the intuitive idea that we can use one empirical model to simulate another one -- either deterministically or more generally using classical, pre-shared randomness to postprocess the measurement results. The key techniques are (i) letting the results of a measurement depend on a \emph{set} of other measurements, (ii) incorporating stochastic postprocessing, using the Kleisli category of the distribution monad. These let us reinterpret various concepts and results from the literature in terms of the resulting category. For example, in Theorem~\ref{thm:globalsectionsasmorphisms} we characterize non-contextual models as those admitting a morphism from the terminal object.

If one thinks of empirical models as concrete experimental set-ups, one can give a concrete interpretation for the morphisms. A map $d\to e$ consists of a method to simulate the data that $e$ would generate using only $d$. We will illustrate this idea with the following story: imagine a lab to perform $e\colon \langle X,\mathcal{M},(O_x)_{x\in X}\rangle$ where for each $x\in X$ there is a graduate student responsible for measuring $x$ and a professor choosing which $C\in\mathcal{M}$ to execute on a given day. One night the students have a party at the lab, and all the equipment breaks down. Fearing the consequences, the students decide that experiments must continue without the professor noticing any abnormalities in the data generated. Luckily, the students happen to have friends at a lab containing an experimental set-up for $d$, and the students make a plan: when a measurement $x$ is called for, the student responsible for it calls each student responsible for a measurement in some set $\pi(x)$ of measurements at the lab for $d$, and then outputs a result $o\in O$ based on the outcome. If the students manage to fool the professor (who knows what the data should look like) no matter which run $C_1,\dots, C_n\in \mathcal{M}$ of experiments she chooses, the protocol defines a map $d\to e$, \ie a way of simulating $e$ using $d$.

However, one gets different notions of a map depending on the powers given to students. In particular, the fabricated outcome for $x$ might depend on the measurement of $\pi(x)$ either deterministically or stochastically, using classical and pre-shared randomness. There are also other powers student might be given that we do not study in this work. For example, the students might be allowed to randomize which measurements they are calling upon based on earlier measurement results, or they might be allowed to use a fixed quantum resource.

\subsection{Related work}
Notions of morphisms between empirical models or measurement scenarios have been defined before, \eg in \cite{wester2017almost} and in \cite{shane2013thesis}. The key features we add to obtain more general morphisms are (i) stochasticity and (ii) morphisms with an underlying simplicial relation rather than just a simplicial function. This additional generality is needed for the main results.

Somewhat similar ideas are developed in \cite{amaral2017noncontextual}, and while not written with the sheaf-theoretic language, a rough comparison is possible: given a box (more or less corresponding to an empirical model), they consider the set of boxes one can obtain by pre- and postprocessing with non-contextual boxes. On one hand, their framework allows for probabilistic preprocessing, which we leave for further work. On the other hand, the preprocessing step allows only for a single measurement to be called upon, whereas we allow the simulation of measurement to depend on multiple measurements.

\section{Background}\label{sec:background}

In this section, we set up some notation and briefly cover the basics of the sheaf-theoretic approach to contextuality. We suggest reading \cite{abramsky_brandenburger_sheaf-theoretic_structure} for a more detailed introduction. The only major differences to usual presentations are that we work over a semifield instead of a semiring (to allow for conditional distributions) and that we recall the definition of the Kleisli category of a monad.

\begin{definition} A (finite) measurement scenario is a tuple $\langle X,\mathcal{M},(O_x)_{x\in X}\rangle$, where $X$ is a finite set of measurements, and $(O_x)_{x\in X}$ is an $X$-indexed family of finite nonempty sets of outcomes, $O_x$ being the outcome set for the measurement $x$. The remaining item $\mathcal{M}$ is a \emph{measurement cover}, \ie an antichain $\mathcal{M}\subset \mathcal{P}(X)$ that covers $X$. The members of $\mathcal{M}$ are called (maximal) measurement contexts.
\end{definition}

The intuition is that $X$ encodes the set of possible measurements on some system. However, it might not be feasible to simultaneously measure any subset $A\subset X$. This is captured by the measurement cover $\mathcal{M}$ -- a subset $A\subset X$ is jointly measurable precisely when it is contained in some maximal measurement context, \ie there is $C\in\mathcal{M}$ such that $A\subset C$. In other words, the jointly measurable subsets form an abstract simplicial complex. We routinely think of the pair $\langle X,\mathcal{M}\rangle$ as a simplicial complex rather than making a distinction between the measurement cover and the simplicial complex generated by it -- after all, a simplicial complex can be given by listing all its faces or only its (inclusion) maximal faces. This geometric viewpoint turns out to be quite handy, so we will recall some terminology needed. 

\begin{definition} Given two simplicial complexes (given by their antichains of maximal faces) $\langle X,\mathcal{M}\rangle $ and $\langle Y,\mathcal{N}\rangle$, a relation $\pi\colon X\to Y$ is a \emph{simplicial relation} if it maps faces to faces, \ie, if for every $C\in\mathcal{M}$ we have $\pi C\subset D$ for some $D\in\mathcal{N}$.
\end{definition}

Next, we recall some terminology that uses the outcome sets. The mapping $U\mapsto \prod_{x\in U}O_x$ induces a functor (in fact a sheaf when $X$ is given the discrete topology) $\mathcal{P}(X)\op\to \cat{Set}$, that we call the event sheaf. We will denote it by $\mathcal{E}_X\colon \mathcal{P}(X)\op\to \cat{Set}$. The action of the functor on an inclusion $U\subset V$ is given the obvious projection: a family $s=(s_x)_{x\in V}$ is mapped to $s|_U:=(s_x)_{x\in U}$. Elements of $\mathcal{E}_X(U)$ are called \emph{$U$-sections} and when $U=X$, they are called \emph{global sections}. To keep the empirical picture in mind, one should think of a $U$-section as an assignment that gives for each measurement in $U$ its outcome.

In the sheaf-theoretic framework, an empirical model is formalised as a family of distributions, each distribution being a distribution over the events in a measurement context. However, one wants some leeway in defining what a `distribution' means, and hence we define them rather generally. Breaking from tradition, we are slightly less general than usual -- we define distributions over semifields rather than semirings, because we want to be able to define conditional distributions. This does not exclude examples of interest that we are aware of, and only Theorem~\ref{thm:Vorobev} uses this restriction. Recall that a semifield is a commutative semiring $(R,+,\cdot,0,1)$ in which non-zero elements have multiplicative inverses. Informally, a semifield is like a field except that subtraction need not be possible, as in the semifield of nonnegative real numbers.

\begin{definition} Let $(R,+,\cdot,0,1)$ be a semifield. A (finite) $R$-distribution on a set $X$ is a function $e\colon X\to R$ such that 
	\begin{itemize} 
			\item $\supp e:=\{x\in X\mid e(x)\neq 0\}$ is finite 
			\item $\sum_{x\in X}  e(x)=1$
	\end{itemize}
For set $X$, the set of $R$-distributions on it is denoted by $D_R(X)$.
The independent product of $R$-distributions $e\in D_R(X),$ and $d\in D_R(Y)$ is the distribution $e\otimes d \in D_R(X\times Y)$ defined by by $(e\otimes d) (x,y)=e(x)d(y)$.
When $e\colon X\times Y\to R$ is an $R$-distribution, we define $\supp_Y e:=\{y\in Y\mid \sum_{x\in X}e(x,y)\neq 0\}$. For $y\in\supp_Y e$ we define the \emph{conditional} $R$-distribution $e(- | y)$ by 
    \begin{equation*} e(x| y):=e(-| y)(x)=e(x,y)/\sum_{x\in X}e(x,y)
    \end{equation*}
Taking $e(y)$ to be a shorthand for $\sum_{x\in X} e(x,y)$, the usual laws concerning conditional distributions hold, namely 
    \begin{equation*} e(x| y)=e(x,y)/e(y)\qquad \text{and}\qquad e(x,y)=e(x | y)e(y).
    \end{equation*} 
\end{definition}
One can think of an $R$-distribution on $x$ as a formal linear combination $\sum_i^n r_ix_i$ where the coefficients $r_i\in R$ sum to $1$. The assignment $X\mapsto D_R(X)$ extends to a functor $D_R\colon\cat{Set}\to\cat{Set}$ by sending a function $f\colon X\to Y$ to the function $D_R(f)\colon D_R(X)\to D_R(Y)$ defined by 
	\begin{equation*} D_R(f)(e)(y)=\sum_{x\in f^{-1}y}e(x)
	\end{equation*}	
written in terms of formal convex combinations, we can equivalently write the action of $D_R$ on functions as 
    \begin{equation*} D_R(f)(\sum_i r_ix_i):=\sum_i r_i f(x_i)
    \end{equation*}
The functor $D_R$ is in fact a monad (see \cite{jacobs2010convexity} and references therein) on \cat{Set}, with unit $\eta\colon\id[\cat{Set}]\to D_R$ defined by 
    \begin{equation*} \eta_X(x)(y)=1x=\begin{cases} 1 &\text{if }x=y \\
    0 &\text{otherwise.}
    \end{cases} \end{equation*}
and multiplication $\mu\colon D_R\circ D_R\to D_R$ defined by ``matrix multiplication'', \ie
    \begin{equation*} \mu(\sum_i^n r_i e_i)(x)=\sum_i^n (r_i e_i(x))
    \end{equation*}

We recall (see \eg \cite[Proposition 4.1.6]{borceuxvol2} for more) 
the definition of a Kleisli category for a monad. 
\begin{definition} Let $(T,\mu,\eta)$ be a monad on a category \cat{C}. The \emph{Kleisli-category of }$T$, denoted by \cat{C_T} has the same objects as \cat{C}, but a morphism $A\to B$ in \cat{C_T} is given by a morphism $A\to TB$ in \cat{C}. The composite of $f\colon A\to TB$ and of $g\colon B\to TC$ is given by $\mu\circ T(g)\circ f\colon A\to TB\to TTC\to TC$. The identity morphism $A\to A$ in \cat{C_T} is given by $\eta_A\colon A\to TA$ in \cat{C}. There is a canonical functor $\cat{C}\to \cat{C_T}$ given by $f\mapsto \eta\circ f$.
\end{definition}

For a semifield $R$, we write \cat{Set_R} instead of \cat{Set_{D_R}} to avoid double subscripts, and we write $F_R\colon Set\to Set_R$ for the canonical functor. We record for future use the fact that taking product distributions behaves well with the functor $D_R$.

\begin{lemma}\label{lem:prodsarenatural}
The map $(e,d)\mapsto e\otimes d$ defines a natural transformation $D_R(-)\times D_R(-)\to D_R(-\times -)$.
\end{lemma}
\begin{proof} 
    Given $R$-distributions $e=\sum_i r_ix_i$ and $d=\sum_j s_jy_j$, we calculate as follows:
    \begin{align*}
        &D_R(f\times g)(e\otimes d)= D_R(f\times g)\sum_{i,j} r_is_j (x_i,y_j)\\
           & =\sum_{i,j} r_is_j (fx_i,gy_j)  =\sum_i r_ifx_i\otimes\sum_j gy_j =(D_R(f)e)\otimes (D_R(g)d) &\qedhere
    \end{align*}
\end{proof}

One routinely composes the functor $D_R\colon \cat{Set}\to\cat{Set}$ with the event sheaf $\mathcal{E}_X$ to obtain a presheaf $D_R\circ \mathcal{E}_X\colon \mathcal{P}(X)\op\to \cat{Set}$, which assigns to each set $U$ of measurements the set of $R$-distributions over $U$-sections. Given an inclusion $U\subset V$, we obtain a map $D_R(\mathcal{E}_X(V))\to D_R(\mathcal{E}_X(U))$, given by marginalization: $e\in D_R(\mathcal{E}_X(V))$ is mapped to $e|_U\in D_R(\mathcal{E}_X(U))$ defined by 

\begin{equation*} e|_U s=\sum_{t\in \mathcal{E}_X(V),t|_u=s}e(t).
\end{equation*}

\begin{definition}
For a semifield $R$, an $R$-empirical model over the measurement scenario $\langle X,\mathcal{M},(O_x)_{x\in X}\rangle$ is given by \emph{a compatible family} $e=(e_C)_{C\in\mathcal{M}}$ of $R$-distributions over $C$-sections. More precisely, each $e_C\in D_R(\mathcal{E}(C))$ and the family being compatible means that the distributions agree on overlaps, \ie $e_C|_{C\cap D}=e_D|_{C\cap D}$ for all $C,D\in\mathcal{M}$. If $U\subseteq C$ for some $C\in \mathcal{M}$, we write $e|_U$ instead of $e_C|_U$: after all, $e|_U$ does not depend on the choice of $C\in\mathcal{M}$. We denote $e$ being an empirical model over the measurement scenario $\langle X,\mathcal{M},(O_x)_{x\in X}\rangle$ by $e:\langle X,\mathcal{M},(O_x)_{x\in X}\rangle$.
 
An $R$-empirical model $e:\langle X,\mathcal{M},(O_x)_{x\in X}\rangle$ is
    \begin{itemize}
        \item \emph{$R$-non-contextual} if there is $d\in D_R(\mathcal{E}(X))$ such that $e_C=d|_C$ for all $C\in\mathcal{M}$, \ie if there exists a distribution on global sections explaining it.
        \item \emph{$R$-contextual} if it is not $R$-non-contextual. 
        \item \emph{strongly $R$-contextual} if there is no global section consistent with it, \ie if there is no $s\in \mathcal{E}(X)$ such that $s|_C\in \supp(e_C)$ for all $C\in\mathcal{M}$
\end{itemize}

Given $R$-empirical models $e^1,\dots ,e^n:\langle X,\mathcal{M},(O_x)_{x\in X}\rangle$ over the same measurement scenario and $r_1,\dots ,r_n\in R$ such that $\sum_{i=1}^n r_i=1$, we define the $R$-empirical model $\sum_{i=1}^n r_i\cdot e^i$ by $(\sum_{i=1}^n r_i\cdot e^i)_C=\sum_{i=1}^n r_i\cdot e^i_C$.

The \emph{non-contextual fraction} of an $\mathbb{R}^+$-empirical model $e$ is defined as the maximal $\lambda\in [0,1]$ such that $e$ decomposes as $e=\lambda e^{NC}+(1-\lambda) e'$ where $e^{NC}$ and $e'$ are empirical models with $e^{NC}$non-contextual. It is denoted by $NCF(e)$, and the \emph{contextual fraction} of $e$ is defined as $CF(e):=1-NCF(e)$.
\end{definition}

Strictly speaking, the above definition defines so-called no-signalling empirical models, and more general empirical models are obtained by dropping the compatibility conditions. Since we only work with no-signalling empirical models, we've chosen to just call them empirical models.

In practice, $R$ is usually one of the following semifields:
\begin{itemize}
    \item The semifield $(\mathbb{R}^{+},+,\cdot,0,1)$ of nonnegative reals, giving rise to finitely supported probability distributions. Usually $\mathbb{R}^{+}$-contextuality is either called probabilistic contextuality or just contextuality.
    \item The field $(\mathbb{R},+,\cdot,0,1)$ of real numbers, giving rise to finitely supported \emph{signed} measures. In \cite{abramsky_brandenburger_sheaf-theoretic_structure} it is shown that being $\mathbb{R}$-non-contextual is equivalent to being no-signalling. 
    \item The semifield $(\mathbb{B},\lor,\land,0,1)$ of Booleans. Now finite $\mathbb{B}$-distributions on $X$ correspond to finite subsets of $X$, and $\mathbb{B}$-contextuality is known as logical contextuality or possibilistic contextuality.
\end{itemize}

Since all non-signalling models over $\mathbb{R}$ are non-contextual and hence not strongly contextual, the notion of strong $\mathbb{R}$-contextuality is somewhat uninteresting. Hence one usually cares about strong contextuality over the positive reals or the booleans, and these notions are in fact equivalent, i.e. an $\mathbb{R}^{+}$-empirical model is strongly contextual iff its possibilistic collapse (defined later) is. Thus in the sequel we only talk about strong contextuality omitting the semifield $R$.

\section{Defining morphisms}\label{sec:definingmorphisms}

Before we give the full formal definitions, let us sketch them using the intuitive story from the introduction as guidance. First of all, any measurement $x\in X$ is assumed to be simulated by doing a set of measurements $\pi(x)\subset Y$. Because the students are supposed to work independently once the professor sets them to work, $\pi(x)$ should only depend on $x$ and not on the set of other measurements performed with it. Hence we get a relation $\pi\colon X\to Y$. For the simulation to be successful, for any $C\in \mathcal{M}$ chosen, one must be able to perform the measurement $\pi(C)$, meaning that $\pi(C)$ has to be a jointly measurable subset of $Y$. In other words, $\pi$ must be a simplicial relation.  

Once $\pi$ has been fixed, how about the rest of the simulation protocol? For any joint outcome $s\in \mathcal{E}_Y(\pi C)$ for the measurements in $\pi(C)$, one must know what outcome in $\mathcal{E}_X(C)$ to output. Because different measurements in $X$ are measured by different students, possibly at different locations, we assume that the students cannot coordinate after each student responsible for measuring $x$ receives the outcomes for measurements in $\pi(x)$. Any coordination (including any possible shared \& classical randomness) must be done before various students are instructed to perform their measurements by the professor. Even worse, each student is just told to perform the measurement $x$ without being revealed the full context containing $x$.

Thus the protocol in fact defines, for any $U\subset X$, a map $\mathcal{E}_Y(\pi U)\to \mathcal{E}_X(U)$ (and if randomness is allowed, this map lives in the Kleisli category \cat{Set_R}). After all, if an arbitrary, not necessarily possible $U$ was chosen and for each $x\in U$ the result for $\pi(x)$ was given, the students would still know which result $O^U$ to output. Moreover, the students acting independently implies that the family of maps  $\mathcal{E}_Y(\pi U)\to \mathcal{E}_X(U)$  has to be a natural transformation: it should not matter whether the professor commissions the experiment $V$ and then throws away some of the results to get a result $s\in \mathcal{E}_X(U)$, or whether the students are told to perform $U$ (thus ignoring any results in $\pi(V)\setminus \pi(U)$) straight away.

Having been suitably motivated, we proceed to the actual definitions and discuss examples.
\begin{definition} Let $\langle X,\mathcal{M},(O_x)_{x\in X}\rangle$ and $\langle Y,\mathcal{N},(P_y)_{y\in Y}\rangle$ be measurement scenarios. A \emph{deterministic morphism} $\langle Y,\mathcal{N},(P_y)_{y\in Y}\rangle\to \langle X,\mathcal{M},(O_x)_{x\in X}\rangle$ consists of: 
   \begin{itemize}
        \item a simplicial relation $\pi\colon X\to Y$;
        \item a natural transformation $\sigma \colon\mathcal{E}_Y(\pi (-))\to \mathcal{E}_X(-)$.
    \end{itemize}
    Given an empirical model $d\colon \langle Y,\mathcal{N},(P_y)_{y\in Y}\rangle$, its pushforward along a deterministic morphism $(\sigma,\pi)$ is the empirical model $\sigma_* d\colon\langle X,\mathcal{M},(O_x)_{x\in X}\rangle$ defined by
            \begin{equation*} \sigma_* d_C=D_R(\sigma_C)(d|_{\pi(C)})
            \end{equation*} 
    Let $e\colon \langle X,\mathcal{M},(O_x)_{x\in X}\rangle$ and $d\colon \langle Y,\mathcal{N},(P_y)_{y\in Y}\rangle$ be empirical models. Then a \emph{deterministic simulation} $d\to e$ consists of a deterministic morphism  $\langle Y,\mathcal{N},(P_y)_{y\in Y}\rangle\to \langle X,\mathcal{M},(O_x)_{x\in X}\rangle$ such that 
            \begin{equation*} e=\sigma_* d
            \end{equation*} 
    The category of $R$-empirical models and deterministic  simulations is denoted by \cat{EmpDet_R}.
\end{definition} 

This definition can be motivated from the mathematical point of view as well, for example, via an analogy to algebraic geometry\footnote{This analogy is not perfect, of course. For one, we define our morphisms to go in the opposite direction compared to the underlying map between simplicial complexes. This is to guarantee that ``points'' of $e$, \ie maps from the terminal object to $e$ correspond to distributions on global sections explaining $e$, but one could have reasonably chosen the opposite convention. Perhaps a more surprising difference is that instead of the direct image sheaf we use the inverse image sheaf. Since in our set-up all topological spaces are finite and discrete, using the (in general more complicated) inverse image functor poses no problems.}. Roughly speaking, a scheme is a pair $(X,\mathcal{O}_X)$, where $X$ is a topological space and $\mathcal{O}_X$ is a sheaf of local rings on it, and a morphism of schemes $(X,\mathcal{O}_X)\to (Y,\mathcal{O}_Y)$ consists of a continuous map $\pi\colon X\to Y$ and of a natural transformation $\mathcal{O}_Y\to \pi_* \mathcal{O}_X$ subject to some axioms. Similarly, a measurement scenario comes with an associated simplicial complex and with a presheaf of outcomes/distributions on the simplices, so it is reasonable to expect that a morphism of measurement scenarios consists of an underlying map between simplicial complexes and of a transformation between the two presheaves (suitably composed with the map of simplicial complexes). Moreover, an empirical model is a compatible family over the presheaf. Hence a morphism of empirical models should be a morphism between the underlying measurement scenarios taking one compatible family to the other.

\begin{remark} We did not restrict ourselves to no-signalling models merely because of disinterest: the definition above would not work as-is without no-signalling: after all, for the pushforward to be well-defined, $d|_{\pi(C)}$ should not depend on which maximal context containing $\pi(C)$ is chosen, as long as there is one.
\end{remark} 

\begin{lemma}\label{lem:gluingdetmaps}[Deterministic morphisms can be glued together] Let $\langle X,\mathcal{M},(O_x)_{x\in X}\rangle$ and $\langle Y,\mathcal{N},(P_y)_{y\in Y}\rangle$ be measurement scenarios and $\pi\colon X\to Y$ a simplicial relation. Let $U_1,\dots U_n$ be a cover of $X$ and assume we are given a compatible family of natural transformations $\sigma^i\colon \mathcal{E}_{Y}(\pi(-))\to \mathcal{E}_{U_i}(-)$, \ie each $\sigma^i$ is a map of presheaves on $U_i$, and the $\sigma^i$ agree on overlaps, meaning that $\sigma^i_{U_i\cap U_j}=\sigma^j_{U_i\cap U_j}$ for every $i,j$. Then there is a unique morphism $(\pi,\sigma)$ of measurement scenarios such that $\sigma$ restricts to $\sigma^i$ on each $\mathcal{P}(U_i)$.
\end{lemma}

\begin{proof} The presheaf $\mathcal{E}_X$ is in fact a sheaf, so this is a consequence of the fact that morphisms of sheaves glue. See \eg \cite[Proposition 2.8.1.]{maclanesheaves} 
\end{proof}

\begin{example} Consider the $\mathbb{R}^+$-empirical model $e:\langle \{x,y\},\{\{x,y\}\},(\{0,1\},\{0,1\})\rangle$ given by $e(x\mapsto 1,y\mapsto 0)=e(x\mapsto 0,y\mapsto 1)=1$ and $e(s)=0$ otherwise, \ie $e$ consists of two perfectly anticorrelated flips of a fair coin. Then all the information about the second coin flip is already (deterministically) present in the first, and this is made precise by the deterministic simulation $(\pi,\sigma)\colon e|_x\to e$ given by $\pi(x)=\pi(y)=x$ and $\sigma_{\{x,y\}}(x\mapsto a)=(x\mapsto a,y\mapsto (1-a))$.
\end{example}

\begin{example}[Restriction] If $e\colon\langle X,\mathcal{M},(O_x)_{x\in X}\rangle$ is an empirical model and $i\colon Y\hookrightarrow X$ is a subset of $X$ with inclusion $i\colon Y\to X$, one gets a restricted empirical model $e|_Y\colon \langle Y,\{C\cap Y|C\in\mathcal{M}\},(O_x)_{x\in Y}\rangle$ defined by $(e|_Y)_{C\cap Y}=e_{C\cap Y}$.  One can obviously simulate the restricted empirical model using the larger one, and this corresponds to the morphism $(i,\id)\colon e\to e|_Y$.
\end{example}

\begin{example}[Coarse-graining] Given family of functions $f=(f_x\colon O_x\to P_X)P$, one can define, for $e\colon\langle X,\mathcal{M},(O_x)_{x\in X}\rangle$ the coarse-graining of $e$ along $f$ as the empirical model $e/f\colon \langle X,\mathcal{M},P\rangle$, defined via $(e/f)_C(s):=\sum_{t\in \mathcal{E}_X(C), f|_C (t)=s}e_C(t)$. Applying $f$ to the outcome lets one simulate $e/f$ using $e$, and this is made precise by the morphism $(\id,\sigma)\colon e\to e/f$, where $\sigma$ is the natural transformation induced by $f$.
\end{example}

\begin{example} We now exhibit a situation where $\pi$ is not a function but a relation. Consider the measurement scenario $\langle \{x,y,z\},\{\{x,y,z\}\},(\{0,1\},\{0,1\})\rangle$ of three jointly measurable binary variables and let $e$ be $\mathbb{R}^+$-empirical model over it corresponding to the variables $x$ and $y$ being independent flips of a fair coin and $z$ being the total number of heads modulo 2. Then there is no deterministic simulation $(\pi,\sigma)\colon e|_{\{x,y\}}\to e$ where $\pi$ is a function, because one cannot infer the value of $z$ knowing the value of only one of $x,y$. Informally, one can clearly simulate $e$ by $e|_{\{x,y\}}$: to know the value of $z$, measure both $x$ and $y$ and calculate, and the values of $x$ and $y$ are given by themselves. This is made formal by the deterministic simulation $(\pi,\sigma)\colon e|_{\{x,y\}}\to e$ where $\pi$ is given by $\pi(z)=\{x,y\}$ and $\pi|_{\{x,y\}}=\id$ and $\sigma_{\{x,y\}}=\id$, $\sigma_{z}(x\mapsto a,y\mapsto b)=z\mapsto (a+b\mod 2)$ (note that by Lemma~\ref{lem:gluingdetmaps} this is enough to specify $\sigma$).
\end{example}

\begin{example}\label{ex:hardy} In \cite{mansfieldfritz2012hardy} the authors show that in any $(2,2,l)$ or $(2,k,2)$ Bell scenario, the occurrence of Hardy's paradox is a necessary and sufficient condition for the model to be possibilistically contextual. The crucial thing to note is that here ``occurrence'' means finding a certain partially filled table inside the empirical model in question. One can make sense of this in our framework by saying that there are several versions of Hardy's paradox -- one for each of the ways of completing the table -- and being possibilistically contextual is equivalent to being able to (deterministically) simulate one of these.
\end{example}

\begin{definition} Let $\langle X,\mathcal{M},(O_x)_{x\in X}\rangle$ and $\langle Y,\mathcal{N},(P_y)_{y\in Y}\rangle$ be measurement scenarios. A \emph{($R$-stochastic) morphism} $\langle Y,\mathcal{N},(P_y)_{y\in Y}\rangle\to \langle X,\mathcal{M},(O_x)_{x\in X}\rangle$ consists of 
    \begin{itemize}
        \item a simplicial relation $\pi\colon X\to Y$;
        \item a natural transformation $\sigma \colon \mathcal{E}_Y(\pi (-))\to D_R\circ \mathcal{E}_X(-)$. 
    \end{itemize}
    The composite of 
        \begin{align*} &(\pi,\sigma)\colon \langle Z,\mathcal{K},(Q_z)_{z\in Z}\rangle\to \langle Y,\mathcal{N},(P_y)_{y\in Y}\rangle \text{ and} \\
        &(\rho,\tau)\colon \langle Y,\mathcal{N},(P_y)_{y\in Y}\rangle\to \langle X,\mathcal{M},(O_x)_{x\in X}\rangle
        \end{align*}
    is given by $(\rho\circ\pi,\mu\circ D_R(\rho)\circ\tau)$ \ie by ordinary composition in the first variable and Kleisli composition in the second.
    Given an empirical model $d\colon \langle Y,\mathcal{N},(P_y)_{y\in Y}\rangle$, its pushforward along a morphism $(\sigma,\pi)$ is the empirical model $\sigma_* d\colon\langle X,\mathcal{M},(O_x)_{x\in X}\rangle$ defined by
            \begin{equation*} \sigma_* d_C=\mu D_R(\sigma_C)(d|_{\pi(C)})
            \end{equation*} 
    Let $e\colon \langle X,\mathcal{M},(O_x)_{x\in X}\rangle$ and $d\colon \langle Y,\mathcal{N},(P_y)_{y\in Y}\rangle$ be empirical models. Then a \emph{simulation} $d\to e$ consists of an $R$-stochastic morphism  $\langle Y,\mathcal{N},(P_y)_{y\in Y}\rangle\to \langle X,\mathcal{M},(O_x)_{x\in X}\rangle$ such that 
            \begin{equation*} e=\sigma_* d
            \end{equation*} 
    The category of $R$-empirical models and simulations is denoted by \cat{Emp_R}.
\end{definition}

\begin{remark} In the previous definition, instead of a natural transformation $\sigma \colon \mathcal{E}_Y (\pi (-))\to D_R\circ \mathcal{E}_X (-)$ we could have equivalently asked for a natural transformation $\sigma\colon F_R\circ \mathcal{E}_Y (\pi (-))\to F_R\circ \mathcal{E}_X(-)$, \ie for a natural transformation between two functors $\mathcal{P}(X)\op\to \cat{Set_R}$. This makes it apparent that composition is associative.
\end{remark}

Given the terminology, one would expect deterministic morphisms (deterministic simulations) to be special cases of morphisms (simulations). Indeed, this is the case -- one just needs to apply the inclusion $F_R\colon\cat{Set}\to\cat{Set_R}$ to the Kleisli category. To be more specific, if $(\pi,\sigma)$ is a deterministic morphism, then $(\pi,F_R \sigma)$ is a morphism. That the same claim holds for simulations boils down to the fact that the pushforward of an empirical model does not depend on whether one uses $(\pi,\sigma)$ or $(\pi,F_R \sigma)$. We note that for a simulation $(\pi,\sigma)\colon d\to e$ one only needs the parts of $d$ that $\pi$ maps to.

\begin{lemma}\label{lem:image} Any simulation $(\pi,\sigma)\colon d\to e$ factors through the restriction $d\to d|_{\pi(X)}$.
\end{lemma}

Unfortunately Lemma~\ref{lem:gluingdetmaps} fails for $R$-stochastic simulations. This is precisely because $D_R\circ \mathcal{E}_X$ fails to be a sheaf. However, one can glue together maps (though not uniquely) along partitions\footnote{This is because $R$-distributions form a ``gleaf'', see~\cite{friz:gleaves}.}. It is straightforward to extend the following lemma to an $n$-ary partition $(U_1,\dots U_n)$.

\begin{lemma}\label{lem:gluingmaps} Let $\langle X,\mathcal{M},(O_x)_{x\in X}\rangle$ and $\langle Y,\mathcal{N},(P_y)_{y\in Y}\rangle$ be measurement scenarios and $\pi\colon X\to Y$ a simplicial relation. If $(U_1,U_2)$ is a partition of $X$ and $\sigma^i\colon \mathcal{E}_Y(\pi(-)\to D_R\circ \mathcal{E}_{U_i}(-)$ is a natural transformation between presheaves on $U_i$ for $i=1,2$, then there is a natural transformation $\sigma^1\otimes\sigma^2\colon \mathcal{E}_Y(\pi (-))\to D_R\circ \mathcal{E}_X(-)$  that restricts to $\sigma_i$ on $\mathcal{P}(U_i)$,  defined by
    \begin{equation*} (s\colon \pi(V)\to P) \mapsto (\sigma^1\otimes\sigma^2)_{V}(s):=\sigma^1_{V\cap U_1}(s|_{\pi (V\cap U_1)})\otimes\sigma^2_{V\cap U_2}(s|_{\pi (V\cap U_2)})
    \end{equation*}
\end{lemma}

\begin{proof} The morphism $\sigma^1\otimes \sigma^2$ factors as a composite of natural transformations 

    \begin{align*}\mathcal{E}_Y (\pi(-))\to \mathcal{E}_Y(\pi(-\cap U_1))\times \mathcal{E}_Y(\pi(-\cap U_2))\to D_R \mathcal{E}_X(-\cap U_1)\times D_R\mathcal{E}_X(-\cap U_2)\to \\ D_R(\mathcal{E}_X(-\cap U_1)\times \mathcal{E}_X(-\cap U_2)) \cong D_R\mathcal{E}_X 
    \end{align*}
    where the first natural transformation is the pairing of the restrictions $\mathcal{E}_Y(\pi(-))\to \mathcal{E}_Y(\pi(-\cap U_i))$, the second is the cartesian product of the natural transformations $\sigma^i$, the third is the independent product distribution map that is natural by Lemma~\ref{lem:prodsarenatural}, and the final isomorphism stems from the fact that $\mathcal{E}_X (-\cap U_1)\times \mathcal{E}_X(-\cap U_2)\cong \mathcal{E}_X$.
\end{proof}

\begin{proposition}\label{convexcombos of maps} Let $(\pi,\sigma^i)\colon d\to e_i$ be simulations, where  each $e_i$ is an empirical model over the same measurement scenario. Then  $(\pi,\sum_{i=1}^n r_i\sigma^i)$ defines a map $d\to \sum_{i=1}^n r_i e_i$, where $\sum_{i=1}^n r_i\sigma^i$ is defined by 
    \begin{equation*}(\sum_{i=1}^n r_i \sigma^i)_U(s):=\sum_{i=1}^n r_i\sigma_U^i(s)\
    \end{equation*}
\end{proposition}

\begin{proof} It is not hard to show that convex combinations are preserved by the restriction maps $D_R\mathcal{E}_X(V)\to D_R\mathcal{E}_X(U)$ induced by inclusions $U\to V$, \ie that $(\sum r_i e_i)|_U=\sum r_i (e_i|_V)$. Since each $\sigma_i$ is natural, this implies that $\sum r_i\sigma_i$ is as well. Hence it remains to see that the pushforward of $d$ along $\sum r_i\sigma_i$ is $\sum r_i e_i$. But this is not hard either: 
    \begin{equation*}(\sum r_i \sigma^i)_* d=\sum r_i(\sigma^{i}_* d)=\sum r_i e_i\qedhere
    \end{equation*}
\end{proof}

\begin{example} This time we discuss the examples more informally, leaving the formal description of the morphisms as an exercise. Consider again two jointly measurable variables. The first one is a flip of a fair coin, whereas the second one depends on the first as follows: if the first coin is heads, the second coin flip is fair, whereas if the first coin is tails, the second coin is flipped with a bias of $2/3$. Clearly knowing just the value of the first coin flip lets you simulate the experiment stochastically -- just flip a coin with the appropriate bias. This defines a simulation $e|_x\to e$. One can go even further and simulate $e$ stochastically without using $e$ at all ---just perform the whole experiment yourself. This ability to simulate $e$ from thin air corresponds to non-contextuality of $e$, as will be seen in Theorem~\ref{thm:globalsectionsasmorphisms}.
\end{example}

\section{Reinterpreting contextuality with morphisms}\label{sec:main}

The category \cat{Emp_R} has a terminal object $1$ given by the unique $R$-empirical model on the empty measurement scenario $\langle \emptyset,\emptyset,\emptyset\rangle$. Admitting a morphism from the terminal object is equivalent to non-contextuality.
\begin{theorem}\label{thm:globalsectionsasmorphisms}
Distributions on global sections explaining $e$ are in one-to-one correspondence with simulations $1\to e$.
In particular, an $R$-empirical model is non-contextual iff there is a simulation $1\to e$.
\end{theorem} 

\begin{proof} Given an $R$-empirical model $e\colon\langle X,\mathcal{M},(O_x)_{x\in X}\rangle$, the data for a simulation $1\to e$ consists of a simplicial relation $X\to\emptyset$, of which there is only one, and of a natural transformation from the constant singleton presheaf to $D_R\circ\mathcal{E}_X$. Hence, the data for a simulation $1\to e$ corresponds to a probability distribution in $D_R(\mathcal{E}_X(X))$. This data is in fact a simulation iff it restricts to $e_C$ for each $C\in\mathcal{M}$.
\end{proof}

Let $f\colon R\to S$ be a homomorphism of semifields. Then $f$ induces a natural transformation (in fact, a monad homomorphism) $\kappa(f)\colon D_R\to D_S$ defined by 
    \[ 
                    e\in D_R(X)\mapsto \kappa(f)_X(e)\in D_S(X) \qquad \qquad \kappa(f)_X(e)(s):=f(e(s))
    \]
This is known to map $R$-empirical models to $S$-empirical models by sending the $R$-empirical model $e\colon \langle X,\mathcal{M},(O_x)_{x\in X}\rangle$ to the $S$-empirical model $F_f(e)\colon \langle X,\mathcal{M},(O_x)_{x\in X}\rangle$ defined by $F_f(e)_C=\kappa(f)(e_C)$. 

\begin{theorem}
The map $e\mapsto F_f(e)$ extends to a functor $\cat{Emp_R}\to\cat{Emp_S}$. Hence $S$-contextuality of $F_f(e)$ implies $R$-contextuality of $e$. In particular, logical contextuality implies probabilistic contextuality.
\end{theorem}

\begin{proof} The action of $F_f$ on an $R$-stochastic simulation $(\pi,\sigma)$ is defined by 
    \begin{equation*}
        F_f(\pi,\sigma):=(\pi,\kappa(f)\circ\sigma)
    \end{equation*}
It is straightforward to check that this defines a functor $\cat{Emp_R}\to\cat{Emp_S}$. By construction it preserves the terminal object, so that if there is a map $1\to e$ in \cat{Emp_R}, then there must be one in \cat{Emp_S}.
\end{proof}

The possibilistic collapse of an $\mathbb{R}^+$-empirical model is defined to be the image of $e$ under the functor $F\colon \cat{Emp_{\mathbb{R}^+}}\to\cat{Emp_\mathbb{B}}$ induced by the unique semiring homomorphism $\mathbb{R}^+\to\mathbb{B}$. It is known not to be surjective on objects~\cite[Proposition 9.1.]{abramsky2013relational} and neither is it full -- for example, there are probabilistically contextual models that are not logically contextual, so that there is a map $1\to F(e)$ in \cat{Emp_\mathbb{B}} not arising as the image of a map $1\to e$ in \cat{Emp_{\mathbb{R}^+}}

While $R$-contextuality admits a categorical description, we have not been able to obtain a similar one for strong contextuality. However, morphisms do respect strong contextuality.

\begin{theorem} Let $(\pi,\sigma)\colon d\to e$ be an $R$-stochastic simulation where $R$ is either $\mathbb{R}^{+}$ or $\mathbb{B}$. If $e$ is strongly contextual, $d$ is as well.
\end{theorem}

\begin{proof} We prove the contrapositive, so assume that $d\colon \langle Y,\mathcal{N},(P_y)_{y\in Y}\rangle$ is not strongly contextual. Then there exists a global section $s\in \mathcal{E}_Y(Y)$ such that $s|_D\in\supp(d_D)$ for all $D\in\mathcal{N}$. Then $e=\sigma_* d$ implies that any global section $p\in \mathcal{E}_X(X)$ in the support of $\sigma_X(s)$ satisfies $p|_C\in\supp e_C$ for every $C\in\mathcal{M}$, so that $e$ is not strongly contextual.
\end{proof}

\begin{lemma}\label{lem:pushfowardofmixtures} Pushforward commutes with convex mixtures: if  $d_1,\dots d_n:\langle X,\mathcal{M},(O_x)_{x\in X}\rangle$ are empirical models, $r_1\dots r_n\in R$ sum to $1$ and $(\pi,\sigma)\colon\langle X,\mathcal{M},(O_x)_{x\in X}\rangle\to \langle Y,\mathcal{N},(P_y)_{y\in Y}\rangle$ is an $R$-stochastic simulation, then 
    \begin{equation*}
        \sigma_* (\sum r_i d_i)=\sum r_i \sigma_* (d_i)
    \end{equation*}
\end{lemma}

\begin{proof}
This can be established via a straightforward calculation, or via abstract nonsense. We opt for the latter: objects in the Kleisli category are the free $D_R$-algebras and hence $\mu D_R(\sigma)$ is a homomorphism of $D_R$-algebras, which means that the equation to be proved holds. 
\end{proof}

\begin{theorem}\label{thm:ncf}
Non-contextual fraction defines a functor $\cat{Emp_{\mathbb{R}^+}}\to [0,1]$, where $[0,1]$ is viewed as a category via its order relation.
\end{theorem}

\begin{proof}
 Let $(\pi,\sigma)\colon d\to e$ be a simulation, and consider any decomposition of $d$ as $\lambda d^{NC}+(1-\lambda)d'$, with $d^{NC}$ non-contextual. Then 
 \begin{equation*}
        e=\sigma_*  d=\sigma_*  (\lambda d^{NC}+(1-\lambda)d')=\lambda\sigma_*  d^{NC}+(1-\lambda)\sigma_* d'
\end{equation*}
where the last equation uses Lemma~\ref{lem:pushfowardofmixtures}. Now $\sigma_*  d^{NC}$ is non-contextual and hence the convex decomposition of $e$ obtained implies that $NCF(d)\leq NCF(e)$.
\end{proof}

\begin{definition}
Let $\langle X,\mathcal{M}\rangle$ be a simplicial complex. Given a vertex $x\in X$ that belongs to exactly one maximal face, we say that there is a \emph{Graham-reduction} from $\langle X,\mathcal{M}\rangle$ to the subcomplex $\langle X\setminus\{x\},\{C\setminus\{x\}\mid C\in \mathcal{M}\}\rangle$. If there is a Graham-collapse from $\langle X,\mathcal{M}\rangle$ to $\langle Y,\mathcal{N}\rangle$, we write $\langle X,\mathcal{M}\rangle\rightsquigarrow \langle Y,\mathcal{N}\rangle$. We call a simplicial complex $\langle X,\mathcal{M}\rangle$ \emph{acyclic} if there is a sequence of Graham collapses from $\langle X,\mathcal{M}\rangle$ to the empty simplicial complex.
\end{definition}

Vorob'evs theorem \cite{vorob1962consistent} is already known (see \cite[p.13]{cohomology_of_contextuality}) to imply that for questions about contextuality one can consider the reduced model. This can be recast and proved in terms of simulations. 

\begin{theorem}\label{thm:Vorobev} Graham-reductions induce morphisms, \ie if $e\colon\langle X,\mathcal{M},(O_x)_{x\in X}\rangle$ is an empirical model and $\langle X,\mathcal{M}\rangle\rightsquigarrow \langle X\setminus\{x\},\{C\setminus\{x\}\mid C\in \mathcal{M}\}\rangle$ is a Graham-reduction, then there is a simulation $e|_{X\setminus\{x\}}\to e$. Hence any empirical model with an acyclic underlying simplicial complex is non-contextual.
\end{theorem}

\begin{proof}
We define a simulation $e |_{X\setminus\{x\}}\to e$ in the appendix using Lemma~\ref{lem:gluingmaps}.    
\end{proof}

Next we define a monoidal structure on \cat{Emp_R}. On objects, the monoidal product (defined in~\cite{contextual_fraction}) sends the pair  
\[(e\colon \langle X,\mathcal{M},(O_x)_{x\in X}\rangle, d\colon \langle Y,\mathcal{N},(P_y)_{y\in Y}\rangle)\] to the empirical model 
\[e\otimes d\colon \langle X\sqcup Y,\mathcal{M}*\mathcal{N}:=\{C\sqcup D\mid C\in\mathcal{M},D\in\mathcal{N}\},(O_x)_{x\in X}\sqcup (P_y)_{y\in Y}\rangle.\]
One should think of this operation as doing experimental set-ups and simulations independently in parallel. Formally, the action on is defined by 
 \[(e\otimes d)_{C\sqcup D}=e_C\otimes d_D \qquad (\pi,\sigma)\otimes (\rho,\tau)=(\pi\sqcup \rho,\sigma\otimes\tau),\]
where $\sigma\otimes\tau$ is the natural transformation built using Lemma~\ref{lem:gluingmaps}. It is straightforward but tedious to check that this results in a symmetric monoidal structure on \cat{Emp_{R}} with the tensor unit given by the terminal object $1\colon\langle \emptyset,\emptyset,\emptyset\rangle$. One might view this SMCs as a resource theory (in the sense of~\cite{coecke2016mathematical}) for contextuality. From this point of view Theorem~\ref{thm:ncf} amounts to saying that the non-contextual fraction is a monotone for our resource theory. Next we will prove that this resource-theory is no-cloning, \ie that no contextual empirical model can be cloned using only classical correlations. 

 \begin{theorem}[No-cloning] There is a simulation $e\to e\otimes e$ iff there is a simulation $1\to e$, \ie if $e$ is non-contextual.
 \end{theorem}

\begin{proof}
    The parallel composition of a morphism $1\to e$ with itself gives rise to the composite \[e\to 1\cong 1\otimes 1\to e\otimes e,\] proving the easier direction. 

    For the converse, consider an empirical model $e\colon \langle X,\mathcal{M},(O_x)_{x\in X}\rangle$ that can be cloned freely, \ie admits a simulation $e\to e\otimes e$. We use induction on $|X|$ to build a simulation $1\to e$, the base case $X=\emptyset$ being clear. Assuming the claim holds for all $m<|X|=:n$, we proceed. Note that from $e\to e\otimes e$ and monoidal operations we can build a simulation $(\pi,\sigma)\colon e\to e^{\otimes n}$. Now $\pi\colon \langle X^{\sqcup n},\mathcal{M}^{*n}\rangle \to \langle X,\mathcal{M}\rangle$. Define $(\pi_i,\sigma^i)\colon e\to e$ as the composite $e\to e^{\otimes n}\to e^{\otimes n}|_{X_i}\cong e$, where the morphism in the middle is the restriction to $i$th copy of $e$ in $e^{\otimes n}$. Now $\pi\colon X^{\sqcup n}\to X$ being simplicial means that $\bigcup_{i=1}^n \pi(C_i)$ is a jointly measurable subset whenever each $C_i$ is. We split into two cases:

    If $\pi_i{X}=X$ for each $i$, then, enumerating $X$ as $x_1,\dots x_n$, we can choose $y_i$ such that $x_i\in \pi_i(y_i)$ for each $i$. Hence $X=\cup \{x_i\}$ is jointly measurable, so that $e$ is non-contextual.

    Otherwise $\pi_i{X}\subsetneq X$ for some $i$. Then $(\pi_i,\sigma^i)$ restricts to a morphism $e|_{\pi_i(X)}\to e$ using Lemma~\ref{lem:image}. Now, doing the restriction $e\to e|_{\pi_i(X)}$ in parallel results in the composite \[e|_{\pi_i(X)}\to e\to e\otimes e\to e|_{\pi_i(X)}\otimes e|_{\pi_i(X)}.\] Hence by the induction assumption $e|_{\pi_i(X)}$ is non-contextual, so take the composite $\emptyset\to e|_{\pi_i(X)}\to e$.
\end{proof}

\section{Further questions}\label{sec:further}

\paragraph{Structure of \cat{Emp_R}} The main thread going through this work is that the existence of a simulation $d\to e$ implies that $e$ is at most as contextual as $d$ is. Assuming that our definition fully captures the intuitive idea of ``simulation with classical randomness'' (see more below), one could argue that the study of the contextuality hierarchy amounts to studying the category \cat{Emp_R}. This motivates cataloguing the categorical properties of \cat{Emp_R}. For instance, does the factorization from Lemma~\ref{lem:image} give rise to a (co)pure-(co)mixed factorization system in the sense of~\cite{cunningham2017purity}?

Likewise, Theorem~\ref{thm:globalsectionsasmorphisms} guarantees that non-contextual models can be defined in terms of the category \cat{Emp_R}. Can other interesting classes of empirical models be recognised in terms of the category \cat{Emp_R}? This question applies most pressingly to strongly contextual models, but can also be asked about models admitting an AvN-argument~\cite{contextuality_cohomology_and_paradox} or a cohomological obstruction~\cite{cohomology_of_contextuality}, or perhaps more boldly, about models admitting a quantum realisation. A weaker variant is to ask if we can transport AvN-arguments or cohomological obstructions along morphisms. Similarly, in \cite{contextual_fraction}, several formulas concerning the contextual fraction are proved, and it would be interesting to see if one can explain those formulas as categorical properties of the functor $CF$.

\paragraph{Relationship with other approaches to contextuality.} There are several alternative approaches contextuality, for example, based on operational equivalence~\cite{spekkens2005contextuality}, hypergraphs~\cite{acin2015combinatorial} or effect-algebras~\cite{sander2015effect}. It would be interesting to relate these formalisms to each other via categorical isomorphisms, equivalences or at least adjunctions. Some work in this direction already exists: for example, \cite{wester2017almost} defines a category of empirical models and a category suitable for the equivalence-based approach, and proves an isomorphism between the two. However, the morphisms on both sides are deterministic\footnote{Stochastic morphisms with preprocessing but without dependence on multiple measurements in the equivalence-based setting are investigated in \cite{duarte2017resource}. Since we do not have preprocessing but have dependencies on several measurements, hoping for an isomorphism might be too much.}. Similarly, in~\cite[Appendix D]{acin2015combinatorial} a bijection between empirical models in our sense and probabilistic models on hypergraphs is defined. However, the question whether this extends to morphisms is not investigated. Moreover, the most obvious definition for morphisms in the combinatorial approach comes from morphisms of hypergraphs, and as such is deterministic. The same drawback applies to the effect-algebraic approach \cite{sander2015effect}: empirical models are regarded as maps from an effect algebra $E$ to the unit interval. In other words, the most natural category to work in is the slice category $\cat{EA/[0,1]}$, but then morphisms of empirical models seem to boil down to deterministically mapping effects to effects (\ie outcomes to outcomes) in a way that preserves the probabilities. Of course, this is not meant as a dismissal of the other approaches: rather, it is a call to define morphisms in the appropriate generality for each of the approaches and to see how the resulting categories relate to each other.

\paragraph{Analogy to complexity theory} The existence of a simulation $d\to e$ seems to be conceptually analogous to the existence of a reduction from one problem another. How far can one push this analogy to complexity theory into ``contextuality theory''? For example, are there complete problems for a ``contextuality-class'', at least if one fixes the measurement scenario?

\paragraph{Stronger notions of morphisms} One might also wonder if stronger notions of a morphism are warranted. For example, if the students are allowed to share quantum resources, one should get a notion of a ``quantum simulation''. The $R$-stochastic simulations use the Kleisli category of the distribution monad. Maybe the notion of a quantum reduction uses something like the quantum monad from \cite{abramsky2017quantummonad}.

Perhaps the most pressing issue is preprocessing: intuitively, when using $d\colon \langle Y,\mathcal{N},(P_y)_{y\in Y}\rangle$ to simulate $e\colon \langle X,\mathcal{M},(O_x)_{x\in X}\rangle$, one should, for a given $x\in X$, be able to use convex mixtures of measurements in $Y$. This might be needed to fully capture the meaning of ``$d$ can be used to simulate $e$ given classical shared correlations''. Indeed, preprocessing is a fundamental part of the maps in \cite{amaral2017noncontextual}, and similarly they are routinely used in \cite{barrett2005nonlocal} and \cite{barrett2005popescu}. In \cite{barrett2005popescu} one considers several copies of PR-boxes and the correlations one can build from them, allowing for simulations where the output of one PR box is fed into another. These two features -- measurement choices depending on classical randomness or on results of other measurements -- result in wider notion of simulation, and we leave formalizing it in the sheaf-theoretic framework as an open question.

One possibility might be to think of a ``simulation $d\to e$ with preprocessing'' as a ordinary simulation $d'\to e$, where $d'$ is built from $d$ and the randomization needed in the preprocessing -- in a sense, $d'$ is the sequential composition of $d$ with a non-contextual model. For this approach to work, one probably needs a good theory of sequential composition of empirical models. The alternative approach is to bite the bullet and let $\pi\colon X\to Y$ be stochastic instead of deterministic. The worry of course being that either the theory does not work or becomes too unwieldy to be useful. In particular, combining both a relation $\pi X\to Y$ and preprocessing might lead to trouble: for instance, what does it mean to mix a measurement $x_1$ with a pair of measurements $x_2$ and $x_3$? Incorporating both of these aspects might be difficult, perhaps for the same reasons combining nondeterministic and probabilistic computing is difficult~\cite{varacca_winskel_2006}.

\section*{Acknowledgements}

I would like to thank Chris Heunen for helpful comments and Rui Soares Barbosa for fixing an issue with the monoidal product. This work was supported by the Osk. Huttunen Foundation.

\bibliographystyle{eptcs}
\bibliography{contextuality}

\appendix 
\section*{\appendixname}
\begin{proof}[Proof of Theorem~\ref{thm:Vorobev}]
    We define a simulation $e |_{X\setminus\{x\}}\to e$. First of all, let $S\in\mathcal{M}$ be the only maximal face containing $x$. We define a simplicial relation $\pi\colon \langle X,\mathcal{M}\rangle\to \langle X\setminus\{x\},\{C\setminus\{x\} | C\in \mathcal{M}\}\rangle$ by
    \begin{equation*} \pi (y)=\begin{cases} C\setminus\{x\} & \text{if }y=x \\
                            \{y\} &\text{otherwise}
                                \end{cases}
    \end{equation*}
Next we define a natural transformation in parts using Lemma~\ref{lem:gluingmaps}. For that purpose, we partition $X$ into $U_1:=\{x\}$ and $U_2:=X\setminus\{x\}$. On $U_2$ the relation $\pi$ restricts to the identity, so we can define $\sigma^2=\eta$, \ie the identity for the Kleisli composition. For the other half, we just need to define a  map $\sigma_x\colon O^{\pi(x)}=O^{C\setminus\{x\}}\to D_R(O^{x})$. Fix $o\in O$. Given $s\colon C\setminus{X}\to O$, we set
    \begin{equation*} \sigma_{x}(s)=\begin{cases} e_C (- | s) & \text{if }s\in\supp_{C\setminus\{x\}}(e_C)\\ 
                                          1\cdot o  & \text{otherwise.}
                                \end{cases}
    \end{equation*}
To check that the resulting morphism $(\pi,\sigma\otimes\eta )$ of measurement scenarios defines a simulation, $e|_{X\setminus\{x\}}\to e$, consider first a context $D\in\mathcal{M}$ that does not contain $x$. Then $D\subset U_2$ and $\sigma\otimes\eta $ restricts to $\eta$, which is the identity for the Kleisli composition.

Consider now the context $C$ containing the deleted point $x$. Now $\pi(C)=C\setminus \{x\}$, and write $e |_{C\setminus\{x\}}=:d$ in the form $\sum_{s\in\supp d} d(s)s$. Then 
    \begin{align*}\mu D_R(\sigma\otimes\eta)_C d&=\mu D_R(\sigma\otimes\eta)_C \sum_{s\in\supp d} d(s)s \\
    &=\mu\sum_{s\in\supp d} d(s)(\sigma\otimes \eta)_C(s) \\
    &=\mu\sum_{s\in\supp d} d(s)(\sigma_{x}(s)\otimes \eta_{C\setminus\{x\}}(s)) \\
    &=\mu \sum_{s\in\supp d} d(s) (e|_C(-|s)\otimes 1s)
    \end{align*}
This being equal to $e_C$ boils down to the fact that $p(x,y)=p(x | y)p(y)$ for conditional $R$-distributions. 
\end{proof}
\end{document}